\RequirePackage[l2tabu,orthodox]{nag}
\documentclass[final,a4paper,abstracton]{scrartcl}
\pdfoutput=1 

\usepackage{todonotes}
\usepackage{nth}
\usepackage[utf8]{inputenc}
\usepackage[T1]{fontenc}
\usepackage{lmodern}    

\usepackage[%
  activate={true,nocompatibility},  
  final,                
  tracking=true,        %
  kerning=true,         %
  spacing=true,         %
  ]{microtype}
\microtypecontext{spacing=nonfrench}

\usepackage{amssymb}
\usepackage{amsmath}

\usepackage[inline]{enumitem}
\setlist[itemize]{label=$\circ$}
\setlist[description]{labelindent=\parindent}
\usepackage[amsmath,hyperref,thmmarks]{ntheorem}

\theoremseparator{.}
\newtheorem{theorem}{Theorem}
\newtheorem{prop}[theorem]{Proposition}
\newtheorem{cor}[theorem]{Corollary}

\newtheorem{definition}{Definition}

\newtheorem{lemma}{Lemma}

\theoremstyle{plain}
\theorembodyfont{\upshape}
\theoremstyle{nonumberplain}
\theoremheaderfont{\itshape}
\theorembodyfont{\upshape}
\theoremsymbol{\ensuremath{_\blacksquare}}
\newtheorem{proof}{Proof}

\theoremstyle{empty}
\renewtheorem{theorem*}{Theorem}
\renewtheorem{lemma*}{Lemma}

\usepackage[pdfusetitle,pdfborder={0 0 0}]{hyperref}

\begin{document}
  
\newcommand{\RR}{\operatorname{Res}}
\newcommand{\PP}{\mathbb{P}}
\newcommand{\Z}{\mathbb{Z}}
\newcommand{\N}{\mathbb{N}}
\newcommand{\R}{\mathbb{R}}
\newcommand{\C}{\mathbb{C}}
\renewcommand{\O}[1]{{\ensuremath{\mathcal{O}(#1)}}}
\newcommand{\softO}[1]{{\ensuremath{\tilde{\mathcal{O}}}(#1)}}
\newcommand{\Res}{\operatorname{Res}}
\newcommand\restr[2]{{
  \left.\kern-\nulldelimiterspace 
  #1 
  \vphantom{\big|} 
  \right|_{#2}
  }}

\title{On the Complexity of Solving Zero-Dimensional\\ Polynomial Systems via Projection}
\author{Cornelius Brand \\ Saarland University\\ Cluster of Excellence (MMCI)\footnote{Part of this work was done while visiting the Simons Institute for the Theory of Computing} \and Michael Sagraloff \\ Max-Planck-Institut f\"ur Informatik, Saarbr\"ucken}
\maketitle
 
\begin{abstract}
Given a zero-dimensional polynomial system consisting of $n$ integer polynomials in $n$ variables, we propose a certified and complete method to compute all complex solutions of the system as well as a corresponding separating linear form $l$ with coefficients of small bit size. For computing $l$, we need to project the solutions into one dimension along $\O{n}$ distinct directions but no further algebraic manipulations. 
The solutions are then directly reconstructed from the considered projections. The first step is deterministic, whereas the second step uses randomization, thus being Las-Vegas.

The theoretical analysis of our approach shows that the overall cost for the two problems considered above is dominated by the cost of carrying out the projections. We also 
give bounds on the bit complexity of our algorithms that are exclusively stated in terms of the number of variables, the total degree and the bitsize of the input polynomials.
\end{abstract}

\section{Introduction}

Let $f_i\in\Z[x_1,\ldots,x_n]$, with $i=1,\ldots,n$, be polynomials of total degree $d_i$ and with integer coefficients of bitsize at most $\tau$, i.e., $f_i$ has \emph{magnitude} $(d_i,\tau)$. We further assume that the system 
\begin{align}\label{def:system}
f_1(x_1,\ldots,x_n)=\cdots=f_n(x_1,\ldots,x_n)=0,
\end{align}
has only finitely many solutions (also "at infinity"). 

There is an extensive literature describing numerous approaches to compute the set $\mathcal{S}$ of complex solutions of (\ref{def:system}), and any attempt to provide a comprehensive overview would go far beyond the scope of this work. Instead, we refer the reader to one of the excellent textbooks~\cite{cattani2005solving,sturmfels2002solving,cox2005using}. 
A well-studied approach based on elimination techniques such as multivariate resultants or Gr\"obner Bases first 
projects the solutions into one dimension and then recovers them from 
the projections. 
That is, given a linear form $l=\sum_{i=1}^n l_i x_i$ with integer coefficients 
$l_i$, we may ask for the image of $\mathcal{S}$ under the mapping $\pi_l:\C^n\mapsto\C$ that sends a point $(x_1,\ldots,x_n)\in\C^n$ to the value $\sum_{i=1}^n l_i x_i\in\C$.  Using elimination techniques, we can compute a univariate polynomial $E^l\in\Z[x]$, which we call an \emph{elimination polynomial along $l$}, such that the set $V(E^l)=\{z\in\C:E^l(z)=0\}$ of roots of $E^l$ 
contains the image $\mathcal{S}^l:=\pi_l(\mathcal{S})$ of $\mathcal{S}$ under $\pi_l$. When reconstructing the solutions from the roots of $E^l$, several problems 
may arise: The set $V(E^l)$ may contain projections of solutions at infinity, so that $V(E^l)\neq \mathcal{S}^l$. This can be resolved by considering a suitable change of coordinates that transforms the corresponding homogeneous polynomial system into a system with only finite solutions. What is even worse, $l$ may be \emph{non-separating} for $\mathcal{S}$, that is, there exist two solutions that map to the same point. In this case, $\pi_{l}$ does not define a bijective mapping between $\mathcal{S}$ and $V(E^l)$, and thus $\mathcal{S}$ cannot be recovered directly from $V(E^l)$. In contrast, if the linear form is known to be separating, then 
efficient methods exist (e.g.~by means of computing a univariate rational representation~\cite{RouillierMultivariate,Alonso1996}) to obtain the solutions from the projections. 

One possible way~\cite{RouillierMultivariate} of computing a \emph{separating linear form} (SLF for short) is to consider a large enough set $\mathcal{L}$ of linear forms, which is known to contain at least one SLF, and to carry out projections along each $l\in \mathcal{L}$ (i.e. we compute $E^l$ and its roots). Then, each linear form $l\in \mathcal{L}$ that maximizes the number of distinct roots of $E^l$ must be separating. For instance, the approach in~\cite{RouillierMultivariate} considers the set
\[
\mathcal{L}:=\{x_1+i\cdot x_2+\cdots+i^{n-1}x_n:0\le i\le (n-1)d^n(d^n-1)/2\},
\]
where $d$ is an upper bound on the degree of all $f_i$. Hence, we need to employ $\Omega(nd^{2n})$ projections along linear forms of bitsize $\O{n^2\log d}$ to compute an SLF, which renders the approach impractical. Our work is driven by the question whether it is possible to compute an SLF using a considerably smaller number of projections. Since two solutions might share $n-1$ coordinates, a reasonable lower bound  for the needed number of projections seems to be $n$. Here, we show that $2n-1$ projections along linear forms of bitsize $\O{n\log d}$ are sufficient, and that the cost for computing an SLF is dominated by the cost for the projections. In addition, the computed linear form has bitsize $\O{n\log d}$, thus being a factor $n$ smaller than what can be obtained with the approach above. 

The main tool underlying our approach is a fast method for the computation of a linear form $l=x+sy$, with $s\in\Z$, that is separating for a two-dimensional grid $G:=X\times Y\subset \C^2$, where $X$ and $Y$ are the sets consisting of the distinct roots of univariate integer polynomials $f$ and $g$ of magnitude $(D,L)$, respectively. In Section~\ref{sec:2dgrids}, we show how to compute such an $s\in\{1,\ldots,D^4\}$, using $\softO{D^3+D^2L}$ bit operations.
This bound is noteworthy as it matches the best bound~\cite{Mehlhorn:roots,DBLP:journals/jsc/Pan02,DBLP:journals/corr/BeckerS0Y15} known for isolating all complex roots of $f$ and $g$, and thus for computing $X$ and $Y$. Notice that using the above result, we may immediately derive the 
current record bound~\cite{DBLP:journals/jc/KobelS15,bouzidi:hal-01114767} of $\softO{d^6+d^5\tau}$ operations for computing an SLF for the solutions of a bivariate system defined by two polynomials of magnitude $(d,\tau)$.
Indeed, using resultant computation we may first project the solutions of this system on both coordinates. Then, the grid $G$ defined as the product of the roots of the two corresponding resultant polynomials (of magnitude $(d^2,\softO{d\tau})$) contains all solutions of the system, and thus an SLF for $G$ also constitutes an SLF for the solutions.  

We extend this approach to compute an SLF for the solutions of a general $n$-dimensional system as given in (\ref{def:system}): We first project the solutions on each of the coordinate axes, which yields sets $X_1$ to $X_n$ in $\C$. Then, the $n$-dimensional grid $G:=X_1\times\cdots\times X_n$ contains all solutions. However, instead of computing an SLF for $G$, we recursively compute SLFs $l'=l[i_1,\ldots,i_k]$ for the canonical embeddings of $\mathcal{S}$ into proper sub-products $X_{i_1}\times\cdots\times X_{i_k}$ of $G$ until we eventually obtain an SLF $l$ for $\mathcal{S}$. This can be achieved by means of a divide-and-conquer strategy, which uses projections along the linear forms $l'$ and our fast method for the computation of an SLF for a two-dimensional grid. 
Our method can be combined with any elimination technique that allows to carry out projections of the solutions along linear forms. The worst case bit complexity of our method is then bounded by
\begin{align}\label{bound:mainresult}
\softO{n\cdot(D^3+D^2L)+n\cdot\Pi},
\end{align}
where $\Pi$ bounds the cost of computing an elimination polynomial for (\ref{def:system}) along a linear form of bitsize $\O{n\log d}$, and $D$ and $L$ constitute bounds on the degrees and the bitsizes of the produced elimination polynomials. 
If a deterministic method is used to compute the elimination polynomials, our method is deterministic as well. 
Using the Las-Vegas algorithm from~\cite{DBLP:journals/jc/EmirisP05,DBLP:journals/jc/Storjohann05} to compute the hidden-variable resultant, we have $D\le d^n$, $L=\softO{(nd)^{n-1}(nd+\tau)}$, and\footnote{\small\label{footnote1}Here, $\omega$ denotes the exponent in the complexity of matrix multiplication. 
In the general case, where each of the considered hidden variable resultants $R(x)$ are obtained from the formula $R(x)=\det M(x)/\det S$ with a non-singular matrix $S$, the bound for $\Pi$ improves by a factor $(nd)^{n-1}$. In this case, the bound in (\ref{mainbound2}) also improves by a factor $(nd)^{n-1}$; see Section~\ref{sec:projection} for more details.} $\Pi=\softO{n^{(n-1)(\omega+1)}(d+\tau)d^{(\omega+2)n-\omega-1}}$. Then, (\ref{bound:mainresult}) writes as 
\begin{align}\label{mainbound2}
\softO{n^{(n-1)(\omega+1)+1}(nd+\tau)d^{(\omega+2)n-\omega-1}},
\end{align}
which bounds the number of bit operations that our algorithm uses in expectation. 
Indeed, within the same complexity, we can even compute $(nd)^{\O{1}}$ different SLFs for the solutions of (\ref{def:system}). With high probability, we may then choose an SLF $l$ in a certified manner such that each root of the corresponding elimination polynomial $E^l$ lifts to a solution of the system. 
Using the (intermediate) separating forms $l'=l[i_1,\ldots,i_k]$ from the computation in the first step, we can finally recover 
all solutions from the projections along $l$. The total cost for this step is also bounded by (\ref{mainbound2}).

The complexity of all steps in our algorithm, except for the computation of the elimination polynomials, is within the best known bound for the computation of the roots of the occurring elimination polynomials. Since the latter bound is suspected to be near-optimal and since any elimination based approach has to compute certain elimination polynomials of comparable magnitude as well as the roots of such polynomials at some point, there is some evidence that our method may perform near-optimal (at least for elimination approaches). 
Note that the bound in (\ref{mainbound2}) is dominated by the bound for the computation of the hidden variable resultant. In particular, for fixed $n$, the cost for the latter task (approximately) scales like $d^{(\omega+2)n}\tau$, whereas the cost for all other steps (approximately) scales like $d^{3n}\tau$. Hence, any improvement on the complexity of computing elimination polynomials yields an improvement of the bound in~(\ref{mainbound2}).

How does our 
bound compare to the complexity results stated in the literature? There has been 
extensive research \cite{Faugere1993329,lazard198177,Heintz1983239,Lakshman1991,Lazard1983,RouillierMultivariate} in the 80s and 90s showing that the computation of multivariate resultants or Gr\"obner Bases as well 
as the computation of the solutions of a zero-dimensional polynomial systems has (arithmetic) complexity 
bounded by $d^{\O{n}}$, thus being polynomial in the size of the dense input representation; see 
also~\cite{hashemi2011sharper} for a more comprehensive overview. There also exist more specific bounds~\cite{mourrain2003accelerated,Canny:1989:SSN:74540.74556,renegar1989worst} yielding an arithmetic complexity for 
computing the solutions of size approximately $\softO{d^{3n}}$. However, we are not aware of any general 
bound on the bit complexity that is comparable to ours, even not for lower-dimensional polynomial system with $3$ or $4$ variables, whereas remarkably, within the last two decades, the thorough investigation~\cite{gonzalez1996improved,Diochnos:2009:APC:1530888.1530918,DBLP:conf/issac/EmeliyanenkoS12,DBLP:journals/jc/KobelS15,bouzidi:hal-01114767} of the (bit) complexity of solving bivariate systems eventually yielded bounds (i.e. $\tilde{O}(d^6+d^5\tau))$ for the computation of an SLF and of all solutions) that are likely to be near-optimal and comparable to our result. 
The method from~\cite{DBLP:journals/jsc/ChengGG12} for solving zero-dimensional system shares some similarities with our approach. There, it is proposed to recursively compute SLFs $l_k=\sum_{i=1}^k l_i x_i$ for the "solutions" of the elimination ideals $\mathcal{I}_k:=\mathcal{I}\cap\mathbb{Q}[x_1,\ldots,x_k]$, where $k=1,\ldots,n$ and $\mathcal{I}:=(f_i)_{i=1,\ldots,n}$ is the ideal defined by the polynomials $f_i$. The crux is that this is done so that all solutions $(\xi,x_{k+1})$ of $\mathcal{I}_{k+1}$ obtained from lifting a specific solution $\xi$ of $\mathcal{I}_k$ project via $l_{k+1}$ into a small neighborhood of $l_k(\xi)$. Following this approach, the solutions of $\mathcal{I}$ can be represented as linear combinations of the roots of univariate polynomials. The method seems to perform well in practice as the actual separation bounds for the roots of the considered elimination polynomials is small compared to the worst-case. However, no complexity analysis is given, and we suspect that the method is not very well suited for a worst-case analysis as it considers the computation of elimination polynomials along linear forms of a very large bitsize (at least in theory).

\section{Preliminaries}\label{sec:projection}
We consider a zero-dimensional polynomial system 
as in (\ref{def:system}). Then, the homogenized system
\begin{align}\label{def:systemp}
F_1(x_1,\ldots,x_{n+1})=\cdots=F_n(x_1,\ldots,x_{n+1})=0,
\end{align}
with $F_i(x_1,\ldots,x_{n+1})\in\Z[x_1,\ldots,x_n]$ a homogenous polynomial of degree $d_i$ and 
\[
F_i(x_1,\ldots,x_n,1)=f_i(x_1,\ldots,x_n)
\] 
has only finitely many solutions in the complex projective $n$-space $\PP^n$. Then, B\'ezout's Theorem says that the total number of solutions in $\PP^n$ is upper bounded by $B:=d_1\cdots d_n\le d^n$.
A solution of the form $(x_1,\ldots,x_n,1)\in\PP^n$ is called \emph{finite},
whereas each solution of the form $(x_1,\ldots,x_n,0)$ is called \emph{infinite}. 
The solution $x_1=\cdots=x_{n+1}=0$ is called \emph{trivial}. Let $\mathcal{S}\subset\C^n$ be the set of all complex solutions of (\ref{def:system}).
Then, the finite solutions $(x_1,\ldots,x_n,1)$ of (\ref{def:systemp}) exactly correspond to the solutions $(x_1,\ldots,x_n)\in\mathcal{S}$ of (\ref{def:system}), whereas the solutions at infinity exactly correspond to the solutions in $\PP^{n-1}$ of the (homogeneous) system
$\bar{F}_1(x_1,\ldots,x_n)=\cdots=\bar{F}_n(x_1,\ldots,x_n)=0$, with $\bar{F}_i:=\restr{F_i}{x_{n+1} = 0}$.

We now briefly review the \emph{hidden variable} approach based on resultant computation, which allows us to project the solutions of (\ref{def:system}) on one of the coordinates; for more details, see~\cite{cox2005using,gelfand2009discriminants}. We may assume that $x_1$ is the coordinate onto which we project.
For a fixed value $x_1=\xi \in \C$, (\ref{def:system}) transforms into 
\begin{align}\label{def:systemfixed}
f_1'(x_2,\ldots,x_n)=\cdots=f_n'(x_2,\ldots,x_n)=0,
\end{align}
with $f_i':=\restr{f_i}{x_1=\xi}$ of generic\footnote{\small For finitely many $\xi$, the degree of $f_i'$ can be smaller than $d_i'$, however, for all other values of $\xi$, each $f_i'$ has degree $d_i'$, which is the degree of $f_i$ considered as a polynomial in the variables $x_2,\ldots,x_n$ with coefficients in $\Z[x_1]$.} degree $d_i'\le d_i$. Let $F_i'\in \C[x_2,\ldots,x_n,x_{n+1}]$ be the corresponding homogenized polynomial of degree $d_i'$, then 
\begin{align}\label{def:systempfixed}
F_1'(x_2,\ldots,x_{n+1})=\cdots=F_n'(x_2,\ldots,x_{n+1})=0
\end{align}
defines a system of $n$ homogeneous polynomials in $n$ variables. It is a well-known fact that there exists a homogeneous polynomial of total degree $D':=\sum_{i=1}^n \prod_{j\neq i} d_i'\le nd^{n-1}$ in the coefficients of the polynomials $F_i'$, the so-called \emph{resultant} $\RR(F_1',\ldots,F_n')$ of the polynomials $F_i'$, which vanishes if and only if the system (\ref{def:systempfixed}) has a non-trivial solution in $\PP^{n-1}$. The resultant is a factor of the determinant of an $m\times m$-matrix $M$, the so-called \emph{Macaulay matrix}, whose entries are given in terms of the coefficients of the polynomials $F_i'$; here, $N:=\sum_{i=1}^n (d_i'-1)+1<nd$ and $m=\binom{N+(n-1)}{n-1}<(nd)^{n-1}$.
Since $f_i'$ has the same coefficients as $F_i'$, one usually defines 
$\RR(f_1',\ldots,f_n') := \RR(F_1',\ldots,F_n')$.

In order to compute the projections of the solutions $\mathcal{S}$ of (\ref{def:system}) onto the first coordinate, we consider $f_i$ as elements of $\Z[x_1][x_2,\ldots,x_n]$ with coefficients in $\Z[x_1]$ of magnitude $(d_i,\tau)$. Hence, $x_1$ is treated as a constant (also "hidden variable"). The \emph{hidden variable resultant} $R^{x_1}=\RR^{x_1}(f_1,\ldots,f_n)$ is a univariate integer polynomial of degree $V$ in $x_1$, with $V\le B$, that vanishes at $x_1=\xi$ if and only if (\ref{def:systempfixed}) has a non-trivial solution in $\PP^{n-1}$. In particular, each solution $(x_1,\ldots,x_n)\in\mathcal{S}$ yields a root $x_1$ of $R^{x_1}$. Hence, the set $V(R^{x_1})$
contains the set $\mathcal{S}^{x_1}=\pi_{x_1}(\mathcal{S})$
of projections of all solutions onto the first coordinate. In general, it is wrong that each root of $R^{x_1}$ also extends to a solution of (\ref{def:system}). However, under certain assumptions, this can be ensured.

\begin{lemma}\label{lem:conditions}
Suppose that (\ref{def:systemp}) has no infinite solution and that
each $f_i$ contains a term of total degree $d_i$ that does not depend on $x_1$. Then, for all $\xi \in \C$, the specialized system (\ref{def:systempfixed}) has no infinite solution. In addition, $V(R^{x_1})=\mathcal{S}^{x_1}$.
\end{lemma}
\begin{proof}
Under the given assumption, we have $d_i'=d_i$, and each polynomial $F_i'(x_2,\ldots,x_n,0)$ is exactly the sum of all terms of the form $c\cdot x_2^{i_2}\cdots x_n^{i_n}$, with $i_2+\cdots+i_n=d_i$ and $c$ a non-zero constant, as they occur in $f_i$. Hence, we have $F_i(0,x_2,\ldots,x_{n},0)=F_i'(x_2,\ldots,x_n,0)$, which shows that each infinite solution of (\ref{def:systempfixed}) extends to an infinite solution of (\ref{def:systemp}). This shows (a).

Now, let $\xi$ be a complex root of $R^{x_1}$. Then, there exists a non-trivial solution in $\PP^{n-1}$ of the specialized system (\ref{def:systempfixed}). From (a), we conclude that (\ref{def:systempfixed}) has no solution at infinity, and thus there exists a solution $(\xi_2,\ldots,\xi_n)\in\C^{n-1}$ of the specialized affine system (\ref{def:systemfixed}). We conclude that $(\xi,\xi_2,\ldots,\xi_n)$ is a solution of (\ref{def:system}). 
\end{proof}

In general, $R^{x_1}$ can be written as $\det M(x_1)/\det S$, where $M(x_1)$ is the Macaulay matrix with entries in $\Z[x_1]$ and $S$ a non-singular square sub-matrix of $M(x_1)$ that does not depend on $x_1$. In the special case, where $S$ is singular, we may use Canny's approach~\cite{DBLP:journals/jsc/Canny90},\cite[\S 4]{cox2005using} 
(known as \emph{Generalized Characteristic Polynomial}) to compute $R^{x_1}$ as the quotient of the trailing coefficients of the (non-zero) characteristic 
polynomials of the matrices $M(x_1)$ and $S$. 
From the bounds on $m$ and $V$, it thus follows that $R^{x_1}$ is of magnitude $(B,\softO{(nd)^{n-1}\tau})$. 
Emiris and Pan~\cite{DBLP:journals/jc/EmirisP05} give a Las-Vegas algorithm to compute
$R^{x_1}$. The main idea underlying their approach is to compute the value of $R^{x_1}$ at $V=\O{B}$ many distinct integer points $x_1=\xi\in\Z$ (each of bit size $\mu=\O{\log B}$), and then to interpolate
$R^{x_1}$ from these values. For computing
$R^{x_1}(\xi) \in \Z$, one evaluates the determinants of $M(\xi)$ and $S$ modulo $p$ for a sufficiently large set of primes (of near-constant bitsize) followed by a Chinese Remaindering step to recover $R^{x_1}(\xi)$. Exploiting that $M$ is quasi-Toeplitz,
the determinants can be computed with $\O{m^2}$ arithmetic operations, which yields the bound $\softO{m^2 n V D' (d+\tau)}=\softO{(d+\tau)n^{2n}d^{4n-3}}$ on the expected costs of computing $R^{x_1}$.
There also exist more adaptive bounds (e.g.~\cite{DBLP:conf/issac/EmirisMT10,DBLP:journals/jc/EmirisP05}) for the magnitude as well as for the complexity of computing the (sparse) resultant that take into account the actual support of the coefficients of the input polynomials (e.g. the \emph{mixed volume}). So for sparse systems, the above bounds constitute significant overestimations. When focusing on general systems, a slightly better bound (with respect to the exponent of $d$) can be derived: Using an asymptotically fast Las-Vegas method~\cite{DBLP:journals/jc/Storjohann05} to compute the determinant of an $m\times m$ matrix with integer entries of bitsize $L$ at expected cost $\softO{m^{\omega}L}$, we obtain the following.

\begin{prop}\label{complexity:resultant}
There is a Las-Vegas algorithm to compute $\det M(x_1)$ and $\det(S)$ in an expected number of 
$$\softO{m^{\omega}(d+\tau)B}=\softO{n^{(n-1)\omega}(d+\tau)d^{(\omega+1)n-\omega}}$$
bit operations. 
If $\det(S)\neq 0$, $\RR^{x_1}_{d_1,\ldots,d_n}(f_1,\ldots,f_n)$ can be computed within the same complexity. Otherwise,
it can be computed in an expected number of bit operations bounded by 
$$\softO{m^{\omega+1}(d+\tau)B}=\softO{n^{(n-1)(\omega+1)}(d+\tau)d^{(\omega+2)n-\omega-1}}.$$
\end{prop}

\begin{proof}
We essentially keep the algorithm from~\cite[Corollary~6.2]{DBLP:journals/jc/EmirisP05} as described above.
That is, we compute the value of $R^{x_1}$ at $\O{B}$ many distinct integer points $x_1=\xi$ of bit size $\O{\log B}$ using determinant computation followed by an interpolation step to recover $R^{x_1}$. 
However, for the determinant computation, we use an asymptotically fast method due to Storjohann~\cite{DBLP:journals/jc/Storjohann05}. The entries of the matrices $M$ and $S$, after specializing $x_1$ to $\xi$, have bit size $\O{\tau+d\cdot\log B}$. Their determinants can be computed (using a Las Vegas algorithm) with $\softO{m^{\omega}(\tau+d\log B)}$ bit operations, where $\omega$ denotes the exponent in the arithmetic complexity of matrix multiplication; recent work~\cite{LeGall:2014:PTF:2608628.2608664} shows that $2\le \omega< 2.3729$. Since we have to carry out these computations for $\O{B}$ distinct values of $x_1$, the claimed bound on the complexity of computing $M(x_1)$ and $S$ follows. If $\det(S)\neq 0$, then $\RR^{x_1}_{d_1,\ldots,d_n}(f_1,\ldots,f_n)=\det M(x_1)/\det(S)$, and we are done. If $\det(S)=0$, then we need to compute the characteristic polynomials $\phi_{M(x_1)}(t)\in\Z[x_1][t]$ and $\phi_S(t)\in\Z[t]$ of $M(x_1)$ and $S$, respectively. For this, we may again consider an interpolation/evaluation approach, where we reduce the computation of the polynomials $\phi_{M(x_1)}(t)$ and $\phi_S(t)$ to the computation of their values at $m$ distinct interpolation points $t=t_i\in\Z$ of small bitsize. This yields an additional factor of size $m=O((nd)^{n-1})$ in the complexity bound.
\end{proof}

Once $R^{x_1}$ is computed, we can use a fast univariate root finder~\cite{Mehlhorn:roots,DBLP:journals/jsc/Pan02,DBLP:journals/corr/BeckerS0Y15} to compute arbitrary small isolating disks for all complex roots of $R^{x_1}$. 

\begin{prop}[Thms. 4 and 5 of \cite{Mehlhorn:roots}] \label{prop:root_comp}
Let $f \in \Z[x]$ be a polynomial of magnitude $(d,\tau)$, and let $\rho$ be an arbitrary positive integer. Then, using $\softO{d^3 + d^2\tau + d\rho}$ bit operations, we can compute a sorted list of isolating disks, each of radius less than $2^{-\rho}$, for all complex roots of $f$.
\end{prop}

To generalize this to projecting solutions along arbitrary directions,
let $l:=l_1\cdot x_1+\cdots+l_n\cdot x_n$ be a linear form with integers $l_i$ of bit size less 
than $\mu$, and let 
$\pi_l:\C^n \rightarrow \C$  be the corresponding mapping.
We say that $\mathcal{S}^l=\pi_l(\mathcal{S})$ is the 
\emph{projection (of the solutions $\mathcal{S}$) along $l$}. 
Suppose that $l_1=1$, then, for computing $\mathcal{S}^l$, we first replace $x_1$ by 
$x_1-l_2 x_2-\cdots-l_n x_n$, yielding
\begin{align}\label{def:systemtr}
f_1^*(x_1,\ldots,x_n)=\cdots=f_n^*(x_1,\ldots,x_n)=0,
\end{align}
with $f_i^*:=f_i(x_1-l_2 x_2-\cdots-l_n x_n,x_2,\ldots,x_n)$. Then, each $f_i$ is an integer polynomial of magnitude 
$(d,\O{d\mu+\tau})$. Let $R^{l}:=\RR^{x_1}(f_1^*,\ldots,f_n^*)$. 
Then, $\mathcal{S}^{l}\subset V(R^l)$. This crucial property of $R^l$ deserves the following definition:

\begin{definition}\label{eliminationpolynomial}
Let $l$ be a linear form as above, then we call $R\in\Z[x]$ an \emph{elimination polynomial for (\ref{def:system}) along $l$} if $\deg R\le B$ and $\mathcal{S}^l\subset V(R)$. We call $R$  \emph{strong} 
if $\mathcal{S}^l= V(R)$. 
\end{definition}

Notice that $R^l$ is a strong elimination polynomial for (\ref{def:system}) along $l$ if both conditions from Lemma \ref{lem:conditions} are fulfilled for the transformed system (\ref{def:systemtr}). Lemma~\ref{lem:direction} shows that, in the case where the linear form
\begin{align*}
&l(\lambda)=x_1+l_2(\lambda) x_2+\cdots+l_n(\lambda) x_n,\text{ with}\\
&l_j:=a_{j0}+a_{j1}\cdot \lambda\in\Z[\lambda]\text{ and }a_{k1}\neq 0\text{ for some }k\in\{2,\ldots,n\}, 
\end{align*}
depends on a parameter $\lambda$, we can always choose $\lambda$ such that (\ref{def:systemtr}) fulfills the second requirement from Lemma~\ref{lem:conditions}.

\begin{lemma}\label{lem:direction}
Let $l(\lambda)$ be a linear form as above, let $\Lambda\subset\Z$ be an arbitrary set of size $|\Lambda|\ge 2nd$, and let $\mu$ be a bound on the bitsize of all $l_j$ and the integers contained in $\Lambda$. 

There exists a Las Vegas algorithm with expected bit complexity $\tilde{O}(d^3(d\mu+\tau)(2d)^{n})$ that
computes an integer $\lambda^*\in\Lambda$ as well as the transformed polynomials $f_i^*=f_i(x_1-l_2(\lambda^*) x_2-\cdots-l_n(\lambda^*) x_n,x_2,\ldots,x_n)$, such that each $f_i^*$ contains a term of degree $d_i$ that does not depend on $x_1$.
\end{lemma}
\begin{proof}
We first prove that at least half of the values in $\Lambda$ yield polynomials $f_i^*$ with the desired property. Let $c_{\bar{\alpha}}\cdot \mathbf{x}^{\bar{\alpha}}$, with $\bar{\alpha}=(\bar{\alpha}_1,\ldots,\bar{\alpha}_n)$, be any term of $f_i$ of total degree $d_i$ that maximizes the degree of $x_1$. We aim to show that, except for at most $\bar{\alpha}_1$ many values of $\lambda$, the polynomial $f_i(x_1-l_2(\lambda)x_2+-\cdots-l_n(\lambda) x_n,x_2,\ldots,x_n)$ contains a term of total degree $d_i$ that is not divisible by $x_1$.
We can write 
\[
f_i=\sum_{j=0}^{d_i} x_1^j\sum_{\alpha'=(\alpha_2,\ldots,\alpha_n):|\alpha'|\le d_i-j}c_{\alpha'}\cdot x_2^{\alpha_2}\cdots x_n^{\alpha_n}
\] 
with constants $c_{\alpha'}$. The coordinate transformation $x_1\mapsto x_1-l_2(\lambda) x_2-\cdots-l_n(\lambda) x_n$ then yields
\[
\sum_{j=0}^{d_i} (x_1-l_2(\lambda) x_2-\cdots-l_n(\lambda) x_n)^j\sum_{\alpha':|\alpha'|\le d_i-j}c_{\alpha'}\cdot x_2^{\alpha_2}\cdots x_n^{\alpha_n},
\] 
and if we restrict to all terms of total degree $d_i$ (in $x_1$ to $x_n$), we obtain 
\begin{align}\label{degreedi}
\sum_{j=0}^{\bar{\alpha}_1} (x_1-l_2(\lambda) x_2-\cdots-l_n(\lambda) x_n)^j\sum_{\alpha':|\alpha'|= d_i-j}c_{\alpha'}\cdot x_2^{\alpha_2}\cdots x_n^{\alpha_n}.
\end{align}
Notice that we only have to sum over all $j$ from $1$ to $\bar{\alpha}_1$ as all other terms must be of total degree less than $d_i$ due to the definition of $\bar{\alpha}_1$. 
Considering the above sum as a polynomial in $x_1$ to $x_n$ with polynomial coefficients in $\lambda$, we can further restrict to those
terms whose coefficient is divisible by $\lambda^{\bar{\alpha}_1}$. This yields
\[
\lambda^{\bar{\alpha}_1}\cdot(-a_{21}x_2-\cdots-a_{n1} x_n)^{\bar{\alpha}_1}\cdot\sum_{\alpha':|\alpha'|= d_i-\bar{\alpha}_1}c_{\alpha'}\cdot x_2^{\alpha_2}\cdots x_n^{\alpha_n},
\] 
which is not identical to zero as $a_{k1}\neq 0$ and $c_{\alpha'}\neq 0$ for $\alpha'=(\bar{\alpha}_2,\ldots,\bar{\alpha}_n)$. From this, we conclude that
there exists at least one term in (\ref{degreedi}) whose coefficient is a non-zero polynomial of degree $\bar{\alpha}_1$ in $\lambda$.
Hence, there exist at most $\alpha_1$ values for $\lambda$ such that $f_i$ does not contain a term of degree $d_i$ that is not divisible by $x_1$. If we apply the same argument to each polynomial  
$f_i$, our first claim follows.

From the above considerations, we conclude that, by choosing a random value from $\Lambda$ yields, with probability at least $1/2$, polynomials $f_i^*$ with the desired properties. Suppose that, for some $\lambda^*$, the polynomials $f_i^*$ are already computed, then we can search for a term in each $f_i^*$ of total degree $d_i$ that is not divisible by $x_1$ for the cost of reading $f_i^*$, which is $\softO{(n+d_i)^n\cdot d_i\mu}$ as $f_i^*$ is a polynomial of magnitude $(d_i,\softO{d\mu+\tau})$. 
It remains to bound the cost for computing the polynomials $f_i^*$. Using Kronecker substitution~(see for example \cite{vzGG13}), we can compute the product of two $n$-variate integer polynomials of magnitude $(d,\tau)$ in $\softO{(d+\tau)(2d)^n}$ bit operations. Hence, can compute all powers $(x_1-l_2 x_2-\cdots-l_n x_n)^j$, for $j=0,\ldots,d$, in $\softO{d(d\mu+\tau)(2d)^n}$ bit operations. Thus, computing all polynomials $f_i^*$ needs $\softO{d^3(d\mu+\tau)(2d)^{n}}$ bit operations. 
\end{proof}

Putting everything together, we obtain the following result:

\begin{cor}\label{cor:main}
Let $l(\lambda)$ and $\Lambda$ be defined as in Lemma~\ref{lem:direction}.
\begin{itemize}
\item[(a)] There is a Las-Vegas algorithm that computes $R^{l(\lambda)}$ in an expected number of bit operations that is bounded by 
$\softO{\Pi(d,\tau,n,\mu)}$, 
where we define $$\Pi(d,\tau,n,\mu):=n^{(n-1)\omega}(d\mu+\tau)d^{(\omega+2)n-\omega-1}.$$
\item[(b)] \label{cor:main:resrootcomp}
Suppose that $R^{l(\lambda)}$ is given. Then, for any $\lambda\in\Lambda$ and $\rho \in \N$, we can compute isolating disks of size less than $2^{-\rho}$ for all roots of $R^{l(\lambda)}$ in 
\begin{align}\label{costforprojection}
\softO{n^{n-1}d^{3n-1}(d\mu+\tau)+d^n\rho}
\end{align}
bit operations.
\item[(c)] \label{cor:main:resrootbound}
For each  $(x_1,\ldots,x_n)\in\mathcal{S}$, we have $2^{-\Gamma}<\max_i |x_i|<2^{\Gamma}$ with $\Gamma:=\max_{i}\log (1+\|R^{x_i}\|_\infty))=\softO{(nd)^{n-1}\tau}$.  
\end{itemize}
\end{cor}
\begin{proof}
(a) follows from Proposition~\ref{complexity:resultant} and the fact that the polynomials $f_i^*$ have bitsize $\softO{d\mu+\tau}$. For (b), we use Proposition~\ref{prop:root_comp} and note that $R^{l(\lambda)}$ has magnitude $(d^n,\softO{(nd)^{n-1}(d\mu+\tau)})$. (c) follows from Cauchy's root bound applied to $R^{x_i}$.
\end{proof}

By flipping $x_1$ and $x_{i}$, the results from the above corollary apply to any linear form $l(\lambda)=x_i+\sum_{j\neq i} l_{j}(\lambda)\cdot x_j\not\equiv x_i$ with $l_j=a_{j0}+a_{j1}\cdot\lambda$. 

\section{Two-Dimensional Grids}\label{sec:2dgrids}
Let $f,g \in \Z[x]$ be two (not necessarily square-free) polynomials of magnitude $(d,\tau)$, and let
 $X =V(f)= \{x_1,\ldots,x_{d'}\}$ and $Y =V(g)= \{y_1,\ldots,y_{d''}\}$ be the corresponding sets of \emph{distinct} complex roots of $f$ and $g$. We further define $G:=X\times Y\subset \C^2$, which is a two-dimensional grid of $d'\cdot d''\le d^2$ many points.
For $s \in \C$, let $l_s: \C \times \C \rightarrow \C,~ (x,y) \mapsto x + sy$.
We call $l_s$ (or simply $s$) \emph{separating for a set $M\subset \C$} if $l_s$ restricted to $M$ is injective, and \emph{non-separating} otherwise. 
Our goal in this section is to show that we can compute an integer (or even a whole sequence of integers) $s$ of bit size $\O{\log n}$ that is separating for $G$ at a cost that is comparable to the computation of the roots of $f$ and $g$.

\begin{theorem} \label{prop:2dim_sep_form}
Let $c$ be a positive integer of size $d^{\O{1}}$. There is an algorithm using $\softO{d^3 + d^2\tau}$ bit operations that outputs $s^* \in S:=\{1,\ldots,(d^4-1)c\}$
such that $l_{s}$ is separating for $G=X \times Y$ for all $s\in S^* = \{s^*,\ldots,s^*+c-1\}$.
Additionally, it holds that 
\[
|(x+sy)-(x'+sy')|\ge \frac{1}{4}\cdot |y-y'|
\]
for any two distinct elements $(x,y)\in G$ and $(x',y')\in G$.
\end{theorem}

In Section~\ref{sec:preliminariesrootbounds}, we fix some definitions and recall well-known (amortized) bounds on the separations and absolute values of the roots of an integer polynomial. Then, in Section~\ref{sec:proofgrid}, we prove the above Theorem. In Section~\ref{sec:preimage}, we show that if $s$ is separating for $G$ and if a subset $G'$ of $G$ maps via $l_s$ one-to-one onto a subset $Z$ of $V(h)$, where $h\in\Z[x]$ is of magnitude $(d,\tau)$, then we can recover $G'$ from $Z$ using $\softO{d^3+d^2\tau}$ bit operations. 

\subsection{Definitions and Bounds}\label{sec:preliminariesrootbounds}

Let $X$ and $Y$ be defined as above. 
For $i,j,k,l \in \{1,\ldots,d\}$, with $i<j$ and $k<l$, we define $\nu_{ij} := |x_i - x_j|$ and $\delta_{kl} := |y_k - y_l|$. Let $\mathrm{N},\Delta$ be the sets of all $\nu_{kl},\delta_{ij}$, respectively. Furthermore, let $\mathcal{F}:=\{\frac{\nu}{\delta}:\nu\in\mathrm{N}\text{ and }\delta\in\Delta\}$. 
Notice that, for the proof of Theorem~\ref{prop:2dim_sep_form}, it suffices to compute a positive integer $s^*\le d^4c-c$ with $|s-\frac{\nu}{\delta}|\ge \frac{1}{4}$ for all $s\in\{s^*,\ldots,s^*+c-1\}$.

For $x \in \C$ and some $\rho\in\N$, we say that $\tilde{x} \in \C$ is an \emph{approximation of absolute (relative) error } $\varepsilon = 2^{-\rho}$ if $|x - \tilde{x}| < \varepsilon$ ($(1-\varepsilon)x < \tilde{x} < (1+\varepsilon)x$). In this case, $\rho$ is called the \emph{absolute (relative) approximation quality} of $\tilde{x}$.
Furthermore, for $x > 0$,
we define $B_x := \log \max\{1,x\} + \log \max\{1,1/x\}$.

We now recall some well-known results on the separations and the absolute values of the roots of an integer polynomial; proofs can be found in~\cite{Mehlhorn:roots,DBLP:journals/jc/KobelS15}. 

\begin{prop}\label{importantbounds}
Let $P\in\Z[x]$ be a polynomial of magnitude $(d,\tau)$ with distinct complex roots $z_1$ to $z_{d'}$ of respective multiplicities $m_i:=\operatorname{mult}(z_i,P)$. Let $\operatorname{Mea}(P):=|\operatorname{lcf}(P)|\cdot \prod_{i=1}^{d'} \max(1,|z_i|)^{m_i}$ be the \emph{Mahler measure} of $P$, and let $\sigma_i:=\operatorname{sep}(z_i,P):=\min_{j\neq i}|z_i-z_j|$ be the \emph{separation of $z_i$}. Then, it holds:
\begin{itemize}
\item[(a)] $\operatorname{Mea}(P)\le \|P\|_{2}\le (n+1)\cdot 2^\tau.$
\item[(b)] $\sum_{i,j:i\neq j}\log\max(|z_i-z_j|,|z_i-z_j|^{-1})=\softO{d^2+d\tau}$.
\item[(c)] $\sum_{i}m_i\log\operatorname{sep}(z_i,P)=\softO{d^2+d\tau}$.
\end{itemize}
\end{prop}

Part (b) of the above proposition implies that
\begin{align}\label{complbound:distances}
\sum\nolimits_{\nu\in \mathrm{N}} B_{\nu}+\sum\nolimits_{\delta\in\Delta} B_{\delta}=\softO{d^2+d\tau}.
\end{align}

\subsection{Separating Forms}\label{sec:proofgrid}
To compute an integer $s^*$ with the properties from Theorem~\ref{prop:2dim_sep_form}, we do not directly work with the set $\mathcal{F}$ of exact fractions but consider instead a set $\tilde{\mathcal{F}}$ of corresponding sufficiently good approximations. We start with the following Lemma:

\begin{lemma} \label{lem:approx}
One can compute approximations of relative quality $\rho$ of $\Delta$ and $\mathrm{N}$ using $\softO{d^3+ d^2\tau+d^2\rho}$ bit operations.
\end{lemma}
\begin{proof}
It suffices to show the claim for $\mathrm{N}$. 
For $\nu=|x_i-x_j|\in \mathrm{N}$, let $s \in \Z$ be such that $\nu = 2^{-s}m$ with $1/2 \leq m < 1$. It follows that $s \leq \max\{0,\log 1/\nu\}$.
Thus, any absolute approximation $\tilde{\nu}$ of quality $B_\nu+\rho$ of $\nu$ constitutes a relative approximation of $\nu$ of quality $\rho$ and of absolute error at most $1$. 
From Proposition \ref{prop:root_comp} (applied to the polynomial $x^2-(a^2+b^2)$ with $a$ and $b$ the real and imaginary part of $x_i-x_j$, respectively), it follows that we can compute such a (dyadic) $\tilde{\nu}$ using $\softO{B_\nu+\rho}$ bit operations. By Proposition~\ref{importantbounds}, $\sum_{\nu\in \mathrm{N}} B_{\nu}=\softO{d^2+d\tau}$, and thus the bit complexity for computing all approximations $\tilde{\nu}$ is bounded by $\softO{d^2+d\tau+d^2\rho}$. 
Notice that the above computation requires an absolute approximation of $x_i$ and $x_j$ of quality $\softO{B_\nu}+\rho$, which is always bounded by $\softO{d^2+d\tau+\rho}$. Hence, using Proposition~\ref{prop:root_comp} (applied to the polynomial $f$), it follows that such approximations can be computed in $\softO{d^3+d^2\tau+d\rho}$ bit operations. This proves our claim.
\end{proof}

In the following, let $\tilde{\mathrm{N}}$ and $\tilde{\Delta}$ be the sets obtained from running the algorithm from Lemma~\ref{lem:approx} with $\rho := \log (64d^4\cdot c)$, where $c$ is a fixed positive integer of size $d^\O{1}$. From the construction of $\tilde{\mathrm{N}}$, it follows that $B_{\tilde{\nu}}=B_{\nu}+\O{\log d}$. In addition, each $\tilde{\nu}\in \tilde{\mathrm{N}}$ is a dyadic number that can be represented by at most $\O{B_\nu+\log d}$ many bits, where $\nu$ is the corresponding exact value contained in $\mathrm{N}$. A corresponding statement also holds for $\Delta$. 
For deriving an SLF for $G$, we employ a kind of binary search on the approximations of the fractions in $\mathcal{F}$. For this, we sort $\tilde{\mathrm{N}}$ using a variant of merge sort. We actually need to modify the classical merge sort algorithm as, in our model of computation, comparisons are not of unit cost, but of cost linear in the bitsize of the operands. This poses a problem if the list to be sorted is composed of two halves of size $\Omega(d^2)$ each, say $L$ and $L'$, such that for all $\ell \in L$ and $\ell' \in L'$, $\ell < \ell'$. In this case, once $L$ and $L'$ are sorted by the respective recursive instance of the sorting procedure, the algorithm continues to compare the largest elements of these sublists with each other. By assumption, no element in $L$ will ever be larger than any element in $L'$, leading to $\Omega(d^2)$ comparisons of the elements in $L$ with the largest element in $L'$. Notice that the largest element in $L'$ might be of bitsize $\Omega(d\tau)$, so comparing it $\Omega(d^2)$ times requires $\Omega(d^3\tau)$ bit operations, which would exceed our claimed complexity bound.

\begin{lemma} \label{lem:sorting}
There is an algorithm sorting $\tilde{\mathrm{N}}$ in $\softO{d^2 + d\tau}$ bit operations.
\end{lemma}
\begin{proof}
In order to prove the desired bound, we alter slightly the merge stage of the algorithm. When merging two sorted sublists $L$ and $L'$, instead of successively comparing the current largest elements $\ell$ and $\ell'$ of $L$ and $L'$, respectively, and inserting them into the merged list accordingly, we perform a binary search on the sublist containing the largest element, say w.l.o.g. that $\ell' > l$, to find the smallest $\ell'' \in L'$ such that $\ell < \ell'' < \ell'$, and insert the part of the list between $\ell''$ and $\ell'$ (including those, of course) into the merged list, followed by $\ell$, and 
carry on with the rest of the algorithm as usual. 
In this case, we say that the binary search in $L'$ was conducted \emph{on behalf of} $\ell$.

We claim about this procedure that it ensures for every element from $\tilde{\mathrm{N}}$ to participate in at most $\O{\log d}$ many comparisons in each merge stage, and hence only in $\O{\log^2 d}$ many comparisons in total.
This will then directly yield the bound on the number of bit operations as stated.

To see this, let $L$ and $L'$ to be the two lists to be merged, and w.l.o.g. let $x \in L = (x_t,\ldots,x,x_k,\ldots,x_1)$ be some element. 
Consider now some comparison of $x$ with an element from $L'$. By the definition of our algorithm, this might either be because there is a binary search conducted in $L'$ on behalf of $x$, or because $x$ is compared in the course of some binary search conducted on behalf of some element of $L'$. In the first case, there are only $\O{\log d}$ comparisons involved, and $x$ is inserted in the merged list thereafter, leading to $\O{\log d}$ comparisons, which is fine. 

On the other hand, $x$ may be compared at most once per binary search conducted from some element $x'$ of $L'$. We can bound the number of binary searches such that this happens as follows: 
If $x_1 > x'$, then the binary search will compare the elements $x_1,x_2,\ldots,x_{2^r}$ to $x'$ in this order, with $r$ minimal with the property that $x_{2^r} < x'$. We may assume that $2^r < 2k$ as, otherwise, $x$ would be removed in this step and inserted in the merged list, preventing it from taking part in any more comparisons. Assume that $x$ is compared to some element, which implies that $2^r \geq k+1$. This means that the sought element cannot be contained in $x_1,\ldots,x_{\lfloor k/2 \rfloor}$, and hence at least $\lfloor k/2 \rfloor$ elements will be removed and inserted into the merged list a consequence of this binary search. Since $k < d$, this bounds the number of times this can happen by $\log d$. 
So, in total, every element is compared at most $\O{\log d}$ times in a single merge stage.
The cost for all comparisons is then upper bounded by $\O{\log^2(d)\cdot \sum_{\nu} (B_{\nu}+\log d)}=\softO{d^2 + d\tau}$, which shows the claim.
\end{proof}

By definition, $\mathcal{F}$ is the image of $\mathrm{N}\times\Delta$ under the mapping $(\nu,\delta) \mapsto \nu/\delta$.
In a similar vein, we will now define a set $\tilde{\mathcal{F}}\subset \R\cup \{+\infty\}$ as the image of $\tilde{\mathrm{N}}\times\tilde{\Delta}$ under a slightly modified mapping $[\cdot]$, which differs from the initial mapping in the way that a pair $\tilde{\phi}=(\tilde{\nu},\tilde{\delta})$ is either mapped to $0$ or $+\infty$ if $\tilde{\nu}\ll\tilde{\delta}$ or $\tilde{\nu}\gg\tilde{\delta}$, respectively.
More precisely, let $e_1,e_2 \in \Z$ be such that $2^{e_1} \leq \tilde{\nu} \leq 2^{e_1 + 1}$ and $2^{e_2} \leq \tilde{\delta} \leq 2^{e_2 + 1}$. If $e_1 + 4 - e_2 \le 0$, in which case $\tilde{\nu}/\tilde{\delta}\le 1/8$, we define $[\tilde{\phi}] := 0$.
If $2^{e_1 - e_2 - 1} \geq 8d^4c$, then $\tilde{\nu}/\tilde{\delta} \geq 8d^4c$, and we define $[\phi] := +\infty$. 
If neither is the case, that is, 
if $1 < 2^{e_1+4-e_2}$ and $2^{e_1 - e_2 - 1} < 8d^4c$, then this implies that $\tilde{\nu}/\tilde{\delta} \in (1/32,32d^4c)$.
In this case, we define $[\tilde{\phi}]$ to be the nearest integer to $\tilde{\nu}/\tilde{\delta}$; we break ties by rounding to the smaller one.
We now collect some properties of this mapping.
\begin{lemma} \label{lem:preprocessing} Let $\phi := (\nu,\delta) \in \mathrm{N}\times\Delta$ and $\tilde{\phi}= (\tilde{\nu},\tilde{\delta})$ be the corresponding approximation in $\tilde{\mathrm{N}}\times\tilde{\Delta}$. Then
\begin{enumerate}
\item[(a)] $[\tilde{\phi}]$ can be computed using $\softO{\min\{B_\delta,B_\nu\}+\log d}$ bit operations.
\item[(b)] If $\nu/\delta \in \{1,\ldots,2d^4c\}$, then $[\tilde{\phi}] = \nu/\delta$
\item[(c)] Let $s \in \N \cup \{+\infty\},$ and let $\tilde{\nu},\tilde{\nu}' \in \tilde{\mathrm{N}}$, $\tilde{\delta} \in \tilde{\Delta}$ with $\tilde{\nu}/\tilde{\delta} < \tilde{\nu}'/\tilde{\delta}$.
Then, $[(\tilde{\nu},\tilde{\delta})] \geq s$ implies $[(\tilde{\nu}',\tilde{\delta})] \geq s$.
\end{enumerate}
\end{lemma}

\begin{proof}
For any $m\in\Z$, computing the sign of $e_1-e_2-m$ can be done using 
\[
\softO{\min\{|e_1|,|e_2|\}+|m|}
\]
bit operations as we need not compute the actual values of $e_1-e_2-m$, but only compare $e_1$ with $e_2+m$ by counting digits until the outcome of the comparison is clear, which happens after $2\min\{|e_1|,|e_2|\}+|m|$ steps. Hence, we can already determine whether $[\tilde{\phi}]\in \{0,\infty\}$ in $\O{\log d+\min\{B_{\tilde{\nu}},B_{\tilde{\delta}}\}}$ bit operations.
For the remaining case, we have ensured that $B_{\tilde{\delta}}$ and $B_{\tilde{\nu}}$ differ by at most $\O{\log d}$, and so the division $\tilde{\nu}/\tilde{\delta}$ can be performed in time $\softO{\log d+\min\{B_{\tilde{\delta}},B_{\tilde{\nu}}\}}$. Computing the nearest integer to this fraction is at most as expensive. This yields the first claim.

For the second claim, notice that the relative error in $\tilde{\mathrm{N}}$ and $\tilde{\Delta}$ (compared to $\mathrm{N}$ and $\Delta$) is bounded by $1/(64d^4c)$ by the choice of $\rho$.
Hence, $\tilde{\nu}/\tilde{\delta}$ is an approximation of $\nu/\delta$ with relative error of at most $1/(16d^4c)$.
In particular, if $\nu/\delta \in [1,4d^4\cdot c]$, then $\tilde{\nu}/\tilde{\delta}$ is an approximation with \emph{absolute} error of at most $1/4$. So, if $\nu/\delta \in \{1,\ldots,2d^4c\}$, then the ball of radius $1/4$ with center $\tilde{\nu}/\tilde{\delta}$ will contain exactly one integer, namely $\nu/\delta$. Since $3/4 \leq \tilde{\nu}/\tilde{\delta} \leq 2d^4c + 1/4$, we have $[\tilde{\phi}]\notin\{\pm\infty\}$, and thus $[\tilde{\phi}]$ must be equal to $\nu/\delta$.
The third claim follows directly from the definition of $[\cdot]$. 
\end{proof}
\begin{lemma} \label{lem:preimage}
For $s,s' \in \{1,\ldots,d^4\cdot c\}$ with $s\le s'$, the cardinality of the preimage $P := P(s,s') \subset \tilde{\mathrm{N}}\times\tilde{\Delta}$ of $\{s,\ldots,s'\}$ under $[\cdot]$ can be computed using $\softO{d^2 + d\tau}$ bit operations.
\end{lemma}
\begin{proof}
By Lemma \ref{lem:sorting}, we may assume $\tilde{\mathrm{N}}$ to be sorted.
For each $\tilde{\delta} \in \tilde{\Delta}$, the cardinality of $P \cap (\tilde{\mathrm{N}}\times\{\tilde{\delta}\})$ can be computed using $\softO{(\log d)(B_\delta+\log d)}$ bit operations: Use two binary searches on $\tilde{\mathrm{N}}$ to find the maximal and minimal elements $\nu^\ast,\nu_\ast \in \tilde{\mathrm{N}}$ with $[(\nu_\ast,\tilde{\delta})] < s \leq s' < [(\nu^\ast,\tilde{\delta})]$. 
The cardinality of $P \cap (\tilde{\mathrm{N}}\times\{\tilde{\delta}\})$ is then the number of elements strictly between $\nu_\ast$ and $\nu^\ast$.
Lemma \ref{lem:preprocessing} implies 
correctness and a bound of $\softO{(\log d)(B_\delta+\log d)}$ bit operations for the binary searches. Summing over all $\tilde{\delta}\in\tilde{\Delta}$ yields the bound. 
\end{proof}

We can now prove Theorem~\ref{prop:2dim_sep_form}:
By Lemmas \ref{lem:approx} and \ref{lem:sorting}, we may already assume that $\mathrm{N}$ and $\Delta$ are approximated by corresponding sets $\tilde{\mathrm{N}}$ and $\tilde{\Delta}$, and that $\tilde{\mathrm{N}}$ is sorted.
As $s \in S=\{1,\ldots,d^4\cdot c\}$ is non-separating for $G=X\times Y$ if and only if $s \in \mathcal{F}$, the task can be reformulated as follows. We need to find $s^* \in S$ such that $\{s^*,\ldots,s^*+c\} \cap \mathcal{F} = \emptyset$. So, instead of working with $\mathcal{F}$ directly, we may replace $\mathcal{F}$ with $\tilde{\mathcal{F}}$ as $S\cap \mathcal{F}\subset S\cap \tilde{\mathcal{F}}$ according to part (b) of Lemma~\ref{lem:preprocessing}.
Using the definition of $\tilde{\mathcal{F}}$, the goal becomes to find $s^* \in S$ such that, for all $\tilde{\phi} \in \tilde{\mathrm{N}}\times\tilde{\Delta}$, we have $[\tilde{\phi}] \notin \{s^*,\ldots,s^*+c-1\}$, or equivalently, $|P(s^*,s^*+c-1)| = 0$. 
To do so, we use a bisection procedure, where we may assume $d^4$ and $c$ to be powers of two, say $2^k= d^4$ and $2^{k'}=c$.
Initially, let $s_0 := 1$ and $s'_0 := d^4c$. Inductively, choose
$(s_{i+1},s'_{i+1}) \in \{(s_i,\theta_i),(\theta_{i}+1,s'_i)\}$ such that $|P(s_{i+1},s'_{i+1})|$ is minimized, where $\theta_i := (s_i + s'_i - 1)/2$ is the center of the set $S_i:=\{s_i,\ldots,s_i'\}$.

By definition, the set $S_i$ contains exactly $d^4c/2^i$ elements for $i = 0,\ldots,k$, and hence $S_k=(s_k,s_k+c-1)$. 
Moreover, for $i = 0,\ldots,k-1$, it holds that $P(s_i,s'_i) = P(s_{i},\theta_{i}) \cup P(\theta_{i} + 1,s'_{i})$, which is a disjoint union. This implies that 
$|P(s_{i},s'_{i})| \leq |P(s_{i-1},s'_{i-1})|/2$ for $i = 1,\ldots,k$.
Since $|P(s_0,s'_0)|$ contains at most $\binom{d^2}{2}<d^4$ elements, we conclude that $P(s_k,\ldots,s_k+c-1)$ is empty, and thus 
each $s\in S^*:=\{s^*,\ldots,s^*-c+1\}$, with $s^*:=s_k$, is separating.

For the bit complexity, notice that there are $k = \O{\log d} $ recursive steps involved in the above approach, so the total bound follows from Lemma \ref{lem:preimage}. 

The claim on the distance to all fractions in $\mathcal{F}$ follows from the fact that $[\cdot]$ maps an element $\tilde{\phi}=(\tilde{\nu},\tilde{\delta})\in\tilde{\mathcal{F}}$ to the nearest integer to $\tilde{\nu}/\tilde{\delta}$ if $1/8<\tilde{\nu}/\tilde{\delta}<8d^4c$. For these elements, the corresponding exact fraction $\nu/\delta$ differs from $\tilde{\nu}/\tilde{\delta}$ by at most $1/4$, and thus $|s-\nu/\delta|>1/4$. For all other $\tilde{\phi}$, we either have $\nu/\delta\le 1/2$ or $\nu/\delta\ge 2d^4c$ as $\tilde{\nu}/\tilde{\delta}$ approximates $\nu/\delta$ with relative error at most $1/4$. Hence, also in this case, $|s-\nu/\delta|>1/4$ for all $s\in S^*$.

\subsection{Lifting Projections}\label{sec:preimage}

\begin{lemma} \label{lem:2dim_reconstruct}
Let $G=X\times Y$ and $s\in \{1,\ldots,d^4c\}$ be separating for $G$ as given in Theorem~\ref{prop:2dim_sep_form}. Let $h\in \Z[x]$ be a polynomial of magnitude $(d,\tau)$ and let $Z =\{z_1,\ldots,z_{d^*}\}\subset V(h)$. 

Suppose that, 
for each $z \in Z$, there exists a pair $(x_z,y_z) \in X \times Y$ such that $l_s(x_z,y_z) = x_z + sy_z = z$. Then, we can compute approximations $(\tilde{x}_z,\tilde{y}_z)$ for all pairs $(x_z,y_z)$ of absolute quality $\rho$ using $\softO{d^3 + d^2\tau + d\rho}$ bit operations. \end{lemma}
\begin{proof}
Let $i_k,j_k$ be such that $x_{i_k} + sy_{j_k} = z_k$.
By assumption, this uniquely defines $i_k,j_k$ for all $k$.
Fix some $k \in \{1,\ldots,d\}$, let $i \geq 0$ and define $L_i := 2^i$.
Then, approximating $z_k\in Z$ and each $y_j\in Y$ with absolute quality $L_i+\O{\log d}$ yields approximations $z^{(i)}_k$ and $y^{(i)}_j$ such that $|z^{(i)}_k - sy^{(i)}_j-(z_k-sy_j)|<2^{-L_i-1}$. Further suppose that $x^{(i)}$ is an approximation of $x\in X$ to an absolute error of $2^{-L_i-1}$, and define 
$X^{(i)}_j := \{x \in X : |z^{(i)}_k + sy^{(i)}_j - x^{(i)}| < 2^{-L_i}\}.$ Furthermore, let $M_0:=Y$ and $M_{i+1} := M_i \cap \{y_j \mid X^{(i+1)}_j \neq \emptyset\}$.

We first show that $X^{(i)}_j = \emptyset$ if $L_i > B_{y_j - y_{j_k}}+2$ and $j \neq j_k$.
Indeed, in this case, Theorem \ref{prop:2dim_sep_form} implies that
$|z_k - sy_i - x_i|  \ge \frac{ |y_j-y_{j_k}|}{4}$.
Thus, $L_i > B_{y_j - y_{j_k}}+2$ implies that
 $x \notin X^{(i)}_j$.
If $j = j_k$, then
$z_k - sy_j - x_l = x_{l_k} - x_l$,
and therefore, when $L_i > B_{x_{l_k} - x_l}$ for all $l$, we have that $X^{(i)}_{j_k} = \{x_{l_k}\}$.

Together, this shows that for an $L_i$ that satisfies both bounds, $M_i$ contains exactly $y_{j_k}$, and $X^{(i)}_{j_k}$ contains exactly $x_{l_k}$.
By definition, this is the preimage of $z_k$ under $l_s$.

This discussion suggests the following procedure:
For all pairs $(k,j)$, compute $M_{i}$ and $X^{(i)}_j$ (for increasing $i$) until both contain exactly one element. From Proposition~\ref{importantbounds}~(c), we conclude that $i = \O{\log d+\log\tau}$ for any pair $(j,k)$.

For the bit complexity of this approach, observe that 
the values $B_{y_j - y_{j_k}}$ and $B_{x_{l_k} - x_l}$ are bounded by $\softO{d^2 + d\tau}$, so we can approximate $X,Y,Z$ with absolute precision $L_i$ using $\softO{d^3 + d^2\tau}$ bit operations by Proposition \ref{prop:root_comp}, and this will suffice for all $L_i$ that are considered before the procedure terminates, by the bound on $i$.
Computing $X^{(i)}_j$ can be done using binary search on the approximations of $X$, hence requiring $\O{L_i\log d}$ bit operations. As we double $L_i$ in every step, $M_i$ has the desired form after $\log(d)\cdot\softO{B_{y_j - y_{j_k}} + B_{\operatorname{sep}(x_{l_k} )}}$ bit operations.

For all pairs $(k,j)$, this yields a number of bit operations bounded by 
\[
\log(d)\cdot\softO{\sum_k \sum_{j \neq j_k} {B_{y_j - y_{j_k}}} + \sum_{k} B_{\operatorname{sep}(x_{l_k},f)}}.
\]
By part (b) of Proposition~\ref{importantbounds}, the first sum is bounded by $\softO{d^2+d\tau}$, and since 
$B_{\operatorname{sep}(x_{l_k},f)}=\softO{d^2+d\tau}$ for each $k$, 
the total sum is bounded by $\softO{d^3 + d^2\tau}$.
Now, approximating $X$ and $Y$ with precision $2^{-\rho}$ using Proposition \ref{prop:root_comp} yields the final claim.
\end{proof}

\section{Polynomial Systems}

\subsection{Computation of a Separating Form}
In what follows, we consider a polynomial system as in \eqref{def:system}, with $f_i\in\Z[x_1,\ldots,x_n]$ polynomials of magnitude $(d,\tau)$. Let $\mathcal{S}\subset\C^n$ be the set of all complex solutions of this system.

Our model of computation will be augmented with an oracle for elimination polynomials as follows:
Given a linear form $l=x_j+l_{j+1}x_{j+1}+\cdots+l_n x_n$ with integer coefficients, the oracle returns an elimination polynomial $E^l\in\Z[x]$ for the system \eqref{def:system} along $l$. We further denote $\Pi$ as an upper bound on the bit complexity of calling the oracle for a linear form of bitsize $\O{n\log d}$.

In Section~\ref{sec:projection}, we have already seen how to realize an oracle for elimination polynomials by means of resultant computation, where $E^l=R^l$ is the hidden variable resultant of the polynomials $f^*_i$ obtained after the coordinate transformation $x_j\mapsto x_j-\sum_{i\neq j}l_i x_i$. However, since there exist also other ways to compute elimination polynomials (e.g. using Gr\"obner Basis), we decided to keep the following considerations as general as possible.

When calling our oracle for $l=x_i$, we obtain the set $X_i:=V(E^{x_i})$, which 
contains the projections of the solutions in $\mathcal{S}$ on the $i$-th coordinate. Thus, we have
$\mathcal{S} \subseteq \mathcal{G}:=X_1 \times \ldots \times X_n$.
For $I = \{i_1,\ldots,i_j\}$, with $1 \leq i_1 < \ldots < i_j \leq n$ and $1\leq j\leq n$, let $\pi_I:\C^n\mapsto\C^j$ be the projection on the coordinates $I = \{i_1,\ldots,i_j\}$. In addition, for any linear form $l: \C^j \rightarrow \C$, we define $l^I := l \circ \pi_I$. That is, if $l=l_1 x_1+\cdots +l_j x_j$, then $l^I=l_1\cdot x_{i_1}+\cdots+l_j\cdot x_{i_j}$. In analogous manner to the two-dimensional case, we say that a linear form $l:\C^j \rightarrow \C$ is separating for a set $M\subset\C^j$ if $l$ restricted to $M$ is injective.

\begin{lemma} \label{lem:extend_linform}
Let 
$l_1,l_2$
be SLFs for $\pi_{I}(\mathcal{S})$ and $\pi_{J}(\mathcal{S})$, respectively, where $I$ and $J$ are disjoint subsets of $\{1,\ldots,n\}$. Let $s$ be separating for $Y_1 \times Y_2 := V(E^{l_{1}^I}) \times V(E^{l_2^{J}})$, with $E^{l_1^I}$ and $E^{l_2^J}$ elimination polynomials along $l_1^I$ and $l_2^J$, respectively. 
Then, the linear form $l_1^I + s\cdot l_2^J$ 
is separating for $\pi_{I \cup J}(\mathcal{S})$.
\end{lemma}
\begin{proof}
This follows directly from the definitions and the choice of $l$ and $s$.
\end{proof}

Following a divide and conquer strategy, we can now recursively compute an SLF for $\mathcal{S}$ starting with the projections of $\mathcal{S}$ on each of the coordinates $x_i$. We give details:
For simplicity, suppose that $n$ is a power of two. Write $X[i,j] := \pi_{\{i,\ldots,j\}}(\mathcal{S})$ 
and consider the complete binary tree with root $X[1,n]$, and each node $X[i,j]$ with $|i - j| \ge 1$ having children $X[i,(i+j-1)/2],X[(i+j+1)/2,j]$. We aim to compute SLFs for the set at the respective node without actually computing this set. 
First, for each $i=1,\ldots,n$, we compute $X_i=V(E^{x_i}) \supset X[i]$
by querying the oracle
for $E^{x_i}$, and then computing its roots.
Then, for $i=1,\ldots,n/2$, Theorem \ref{prop:2dim_sep_form} yields an SLF $x+s_iy$, with $s_i \leq d^{4n}$, for $X_{2i-1}\times X_{2i}$, and thus also for $X[2i-1,2i]$.
For the inductive step, assume we can compute SLFs for the sets at all nodes in the levels $1$ to $j$ of the tree, and consider some node on level $j+1$, say w.l.o.g $X[1,2^{j}]$ with children $X[1,2^{j-1}]$ and $X[2^{j-1}+1,2^j]$.
Let $l_1=\sum_{i=1}^{2^{j-1}} l_i' x_i$ and $l_2=\sum_{i=1}^{2^{j-1}} l_i'' x_{i}$ be SLFs for these sets, respectively, and suppose their coefficients have absolute values bounded by $d^{4n(j-1)}$.
We obtain $E^{l_1^I}$ and $E^{l_2^J}$ by calling the oracle twice, 
where $I=\{1,\ldots,2^{j-1}\}$ and $J=\{2^{j-1}+1,\ldots,2^j\}$. 
Again, Theorem \ref{prop:2dim_sep_form} yields a separating form $x+sy$ for 
$V(E^{l_1^I})\times V(E^{l_2^J})$.
By Lemma~\ref{lem:extend_linform}, the linear form $l=\sum\nolimits_{i=1}^{2^{j-1}} l_i' x_i+s\cdot \sum\nolimits_{i=1}^{2^{j-1}} l''_{i} x_{2^{j-1}+i}$ is separating for $X[1,2^j]$.
Since the absolute values of the coefficients increase by a factor 
of  
at most $d^{4n}$, $l$ has coefficients of absolute value at most $d^{4nj}$. Hence, after $\log n$ recursive steps, we obtain am SLF for $\mathcal{S}=X[1,n]$ with coefficients of absolute value $d^{4n\log n}$ or less.

\begin{theorem}
The above algorithm computes an SLF for $\mathcal{S}$ with integer coefficients bounded by $d^{4n\log n}$ using
\[
\softO{n(D^3+D^2L)+n\cdot\Pi},
\]
bit operations,
where $(D,L)$ is an upper bound on the magnitude of all elimination polynomials produced by the algorithm, and $\Pi$ is an upper bound on the bit complexity of calling the oracle for a linear form $l$ of bitsize $\O{n\log d}$.
The algorithm is deterministic if the oracle is deterministic. 

If we use the Las Vegas method from Section~\ref{sec:projection} for resultant computation to realize the oracle, then the above bound transforms into the bound (\ref{mainbound2}) from the introduction.
\end{theorem}

\begin{proof}
The oracle is called $2n-1$ times, and Theorem~\ref{prop:2dim_sep_form} is invoked $n-1$ times. 
For the second claim, 
suppose that the oracle is realized by means of a resultant computation as proposed in Section~\ref{sec:projection}. Then, using Corollary~\ref{cor:main} (a), 
we see that each computation of an elimination polynomial $E^l$ needs 
$\Pi=\softO{n^{(n-1)(\omega+1)}(nd+\tau)d^{(\omega+2)n-\omega-1}}$
bit operations in expectation as $l$ has bitsize $\O{n\log d}$. In addition, the magnitude of each elimination polynomial is bounded by $(d^n,\softO{(nd)^{n-1}(nd+\tau)})$, which shows the second claim.
\end{proof}

We can also slightly modify the above algorithm to compute a sufficiently large set of SLFs from which we can then choose a linear form such that $R^l$ is a strong elimination polynomial. For this, we assume that the oracle is realized by means of resultant computation, that is, we have $E^l=R^l$. Now, suppose that SLFs $l_1=\sum_{i=1}^{n/2} l_i' x_i$ and $l_2=\sum_{i=1}^{n/2} l_i'' x_{i}$ for the sets $X[1,n/2]$ and $X[n/2+1,n]$ are computed as above. In addition, let $E^{l_1^I}$ and $E^{l_2^J}$ be the elimination polynomials along $l_1^I$ and $l_2^J$, with $I=\{1,\ldots,n/2\}$ and $J=\{n/2+1,\ldots,n\}$. By Theorem~\ref{prop:2dim_sep_form}, we can compute a set $S^*:=\{s^*,\ldots,s^*+2nd\}\subset\{1,\ldots,2nd\cdot d^{4n}\}$ such that $x+sy$ is separating for $V(E^{l_1^I})\times V(E^{l_2^J})$ for all $s\in S^*$. This costs at most $\tilde{O}(d^{3n}+d^{2n}(nd)^{n-1}(nd+\tau))$ bit operations, and the linear form $l(s)=\sum\nolimits_{i=1}^{2^{j-1}} l_i' x_i+s\cdot \sum\nolimits_{i=1}^{2^{j-1}} l''_{i} x_{2^{j-1}+i}$ is separating for $\mathcal{S}$ for each $s\in S^*$. We conclude:

\begin{theorem}
There is a Las Vegas algorithm with expected bit complexity \eqref{mainbound2}
to compute a set $S^*:=\{s^*,\ldots,s^*+2nd\}\subset\{1,\ldots,2nd\cdot d^{4n}\}$ and a linear form $l(s)=\sum_{i=1}^n (a_{0i}+a_{1i}s)\cdot x_i\in\Z[s]$ of bitsize $\softO{n\log d}$, such that 
$l(s)$ is separating for all $s\in S^*$.
\end{theorem}
  
Combining Lemma~\ref{lem:conditions} and Lemma~\ref{lem:direction} directly yields the following result:	
	
\begin{cor}
Suppose that the system (\ref{def:system}) has no solution at infinity. Then, there is a Las Vegas algorithm with expected bit complexity 
(\ref{mainbound2}) to compute an SLF $l$ with coefficients of bitsize $\O{n\log d}$ and the corresponding hidden-variable resultant $R^l$, such that $V(R^l)=\mathcal{S}^l$.
\end{cor}

\subsection{Computing the Solutions}
We first consider the case where (\ref{def:system}) has no solution at infinity. By the last subsection, we may assume that, for $j=1,\ldots,\log n$ and $k=0,\ldots,n/2^j-1$, we have already computed 
SLFs $l_{j,k}=\sum_{i=1}^{2^j}l^{(j,k)}_i x_i$ for the sets $X[k\cdot 2^{j}+1,(k+1)\cdot 2^{j}]$, 
respectively. We may further assume that $l=l_{\log n,0}$ is separating for the solutions of our system and that 
$R^l$ is a strong elimination polynomial. Notice that $l^{(j,k)}_1=1$ for all $(j,k)$ and each coefficient 
$l^{(j,k)}_i$ is an integer of bit size $\softO{n\log d}$. Due to the construction of the $l_{j,k}$'s, it holds that
\begin{align}\label{formula recursive}
l_{j,k}(\mathbf{x},\mathbf{y}) = l_{j-1,2k+1}(\mathbf{x}) + s_{j,k}\cdot l_{j-1,2k+2}(\mathbf{y}).
\end{align}
with integers $s_{j,k}$ of bitsize $\softO{n\log d}$, $\mathbf{x}=(x_1,\ldots,x_{2^{j-1}})$ and $\mathbf{y}=(y_1,\ldots,y_{2^{j-1}})$.
Let 
\[
\phi_{j,k}:(x_1,\ldots,x_n)\mapsto \sum\nolimits_{i=1}^{2^j}l^{(j,k)}_i x_{2^j\cdot k+i}
\] be the mapping induced by the linear form $l_{j,k}$, that is, $\phi_{j,k}$ only operates on the variables $x_{2^jk+1}$ to $x_{2^{j+1}}$. We further define $\phi_{0,k}:(x_1,\ldots,x_n)\mapsto x_{k+1}$ as the projection onto the $k+1$-th coordinate for all $k$, and
\[
\phi_j:=\phi_{j,0}\times\cdots\times\phi_{j,n/2^j-1}
\]
as the cartesian product of all $\phi_{j,k}$ for a fixed $j$.

Now, we recursively apply Lemma \ref{lem:2dim_reconstruct} to compute the image $\Phi_j:=\phi_j(\mathcal{S})$ of $\mathcal{S}$ under $\phi_j$.
Notice that $\mathcal{G}=\Phi_{0}$ and $V(R^l)=\mathcal{S}^l=\Phi_{\log n}$ as $R^l$ is a strong elimination polynomial.  
Suppose that $\Phi_j$ is already computed for some $j$, in particular, we know the image $\Phi_{j,k}=\phi_{j,k}(\mathcal{S})$ of $\mathcal{S}$ under the mapping $\phi_{j,k}$. Further notice that the mapping $\phi:(x,y)\mapsto x+s_{j,k}y$ is injective on the product $\Phi_{j-1,2k+1}\times\Phi_{j-1,2k+2}$, and that it maps the image of $\mathcal{S}$ under $\phi_{j-1,2k+1}\times\phi_{j-1,2k+2}$ one-to-one onto $\Phi_{j,k}$.
Hence, using Lemma \ref{lem:2dim_reconstruct}, we may compute the inverse of each point in $\Phi_{j,k}$ under the mapping $\phi$, which yields $(\phi_{j-1,2k+1}\times\phi_{j-1,2k+2})(\mathcal{S})$.
Thus, after $\log n$ recursive steps, we obtain $\mathcal{S}$.

\begin{theorem} \label{thm:approx_noinf}
If the system \eqref{def:system} has no solutions at infinity, then the above algorithm computes approximations (in terms of isolating regions) of absolute quality $\rho$ of $\mathcal{S}$
using 
$$\softO{n^{(n-1)(\omega+1)+1}(nd+\tau)d^{(\omega+2)n-\omega-1}+n d^n\rho}$$
bit operations in expectation.
\end{theorem}
\begin{proof}
There are $\log n$ levels to be considered,
so we employ Lemma~\ref{lem:2dim_reconstruct} at most $n$ times. The involved polynomials are elimination polynomials along the linear forms $l_{j,k}$, which are of 
magnitude $(d^n,\softO{(nd)^{n-1}(nd + \tau)})$. This yields the bound from (\ref{mainbound2}) for reconstructing all solutions of the given system. We can now compute absolute approximations of quality $\rho$ of these solutions by computing corresponding approximations of the roots of the polynomials $R^{x_i}$. Hence, the claimed bound follows directly from Corollary~\ref{cor:main} (b).
\end{proof}

We now remove the condition on the input system to have no infinite solution. We can easily check whether this condition is fulfilled. Namely, (\ref{def:systemp}) has no infinite solution if and only if $\RR(\bar{F}_1,\ldots,\bar{F}_n)\neq 0$. The following Lemma shows that this can be achieved, with probability at least $1/2$, by means of a coordinate transformation.

\begin{lemma} \label{lem:remove_infinity}
Let $\lambda_1$ to $\lambda_n$ be a randomly chosen non-negative integers with $\lambda_i \le 2d^n$ for all $i$. Then, with probability at least $1/2$, the transformed system 
\begin{align}\label{def:systempstar}
F_1^*(x_1,\ldots,x_{n+1})=\cdots=F_{n}^*(x_1,\ldots,x_{n+1}) = 0,
\end{align}
with $F_i^*(x_1,\ldots,x_{n+1})=F_i(x_1,\ldots,x_n,x_{n+1}+\lambda_1 x_1+\cdots \lambda_n x_n)$, has no solution at infinity. 
There is a Las-Vegas algorithm to compute such $\lambda_i$'s and the polynomials $F_i^*$ with expected bit complexity bounded by (\ref{mainbound2}).
\end{lemma}
\begin{proof}
 Let $(x_1,\ldots,x_n,x_{n+1})\in\PP^n$ be an arbitrary non-trivial solution of (\ref{def:systemp}). 
 
If $(0,\ldots,0,1)$ is a solution of (\ref{def:systemp}), then this solution is again mapped to $(0,\ldots,0,1)$ via the coordinate transformation $x_{n+1}+\lambda_1 x_1+\cdots \lambda_n x_n=0$ no matter how we choose the $\lambda_i$'s. Hence, we may assume that there exists a $j\neq n+1$ with $x_j\neq 0$. Then, after fixing $\lambda_i$ for $i\neq j$, there exists at most one value for $\lambda_j$ with $x_{n+1}+\lambda_1 x_1+\cdots \lambda_n x_n=0$. Thus, with probability at least $1-1/(2d^n)$, the solution is mapped to a finite solution of (\ref{def:systempstar}). Since the total number of solutions is bounded by the B\'ezout number $B\le d^n$, the first claim follows.

For the second claim, notice that, after choosing $\lambda_i$'s at random, we can first compute the polynomials $F_i^*$ 
and then compute the resultant $\RR(\bar{F_1^*},\ldots,\bar{F_n^*})$ in order to check whether there is a solution at infinity. Similar as in the proof of Lemma~\ref{lem:direction}, we can bound the cost for computing the polynomials $F_i$ by $\softO{d^3(d\cdot n+\tau)(2d)^n}$ bit operations. Each $F_i^*$ has magnitude $(d,\softO{dn+\tau})$, and thus computing the resultant needs $\softO{(dn)^{\omega\cdot(n-1)}(dn+\tau)}$ bit operations; see also the proof of Proposition~\ref{complexity:resultant}.
\end{proof}

Using this Lemma, we can first transform (\ref{def:system}) into a system (\ref{def:systempstar}) without roots at infinity. Then, we can compute all solutions of (\ref{def:systempstar}) and recover the solutions of (\ref{def:system}) via the backward transformation $x_{n+1}\mapsto x_{n+1}-\sum_{i=1}^n \lambda_i\cdot x_i$. That is, each solution $\mathbf{x^*}=(x_1^*,\ldots,x_n^*,1)\in\mathcal{S}^*$ 
of (\ref{def:systempstar}) maps to a solution $\mathbf{x}=(x_1,\ldots,x_n,x_{n+1})=(x_1^*,\ldots,x_n^*,x_{n+1}^*-\sum_{i=1}^n \lambda_i x_i^*)$ of the initial system. Using only approximate computation, we cannot directly show that $x_{n+1}=x_{n+1}^*-\sum_{i=1}^n \lambda_i x_i^*$ is equal to zero, and thus a solution $\mathbf{x}$ at infinity cannot directly be verified as such. However, by increasing the precision, we either obtain that $x_n\neq 0$ or we may conclude that $|x_i/x_n|$ is larger than the bound from Corollary~\ref{cor:main} (b) on the absolute value of a solution of our system. In the first case, $\mathbf{x}$ is a finite solution, whereas in the second case, $\mathbf{x}$ is an infinite solution. 
A simple analysis of this approach yields the following result.
\begin{theorem}\label{thm:final}
There exists a Las-Vegas algorithm to compute isolating regions of all solutions of (\ref{def:system}) whose cost in expectation is bounded by (\ref{mainbound2}).
\end{theorem}
\begin{proof}
By Lemma~\ref{lem:remove_infinity}, we can first transform our input system into a system (\ref{def:systempstar}) whose solutions are all 
finite, and then use Theorem~\ref{thm:approx_noinf} to compute the set $\mathcal{S}^*$ of its solutions up to an absolute error of $2^{-L}$ in each coordinate for some $L$. Call these approximations $(\tilde{x}_1^*,\ldots,\tilde{x}_n^*,1)$.

By definition, the transformation $x_{n+1}\mapsto x_{n+1}-\sum_{i=1}^n \lambda_i\cdot x_i$ maps each solution $(x_1^*,\ldots,x_n^*,1)\in\mathcal{S}^*$ 
of the transformed system to a solution 
\[
(x_1,\ldots,x_n,x_{n+1})=(x_1^*,\ldots,x_n^*,x_{n+1}^*-\sum_{i=1}^n \lambda_i x_i^*)
\] 
of the initial system. 
Applying this transformation to the approximations of $\mathcal{S}^*$ allows for approximating $x_i$ to an absolute error less than $2^{-L'}$, where $L'\le L-\log n-n\log d$. If 
$|\tilde{x}_{n+1}|>2^{-L'}$, where $\tilde{x_i}$ is the corresponding approximation of $x_i$, we may conclude that $x_{n+1}\neq 0$, and thus $(x_1^*/x^*_{n+1},\ldots,x^*_{n}/x^*_{n+1})$ is a solution of (\ref{def:system}). 

It might happen that the backward coordinate transformation sends a solution of \eqref{def:systempstar} to an infinite solution of \eqref{def:system}.
Namely, this is the case if and only if $x_{n+1}^*-\sum_{i=1}^n \lambda_i x_i^*=0$. Of course, this test for equality cannot be done directly using approximate arithmetic.
However, if $x_{n+1}=x_{n+1}^*-\sum_{i=1}^n \lambda_i x_i^*\neq 0$, then $(x_1^*/x^*_{n+1},\ldots,x^*_{n}/x^*_{n+1})$ is a solution of (\ref{def:system}). By part (b) of Corollary~\ref{cor:main}, we either have $x_i^*/x_{n+1}^*=0$ or $2^{-\Gamma}<|x_i^*/x_{n+1}^*|< 2^{\Gamma}$ with some $\Gamma$ of size $\softO{(nd)^{n-1}\tau}$.  
Hence, choosing $L$ large enough, that is, $L>2\Gamma+\log n+n\log d=\softO{(nd)^{n-1}\tau}$, we either have $|\tilde{x}_{n+1}|>2^{-L'}$, or we may conclude that $x_{n+1}^*-\sum_{i=1}^n \lambda_i x_i^*=0$.

The claim of the number of bit operations of this procedure follows directly from Lemma~\ref{lem:remove_infinity} and Theorem~\ref{thm:approx_noinf} together with the bound on $\Gamma$.
\end{proof}

\bibliographystyle{plain}
\bibliography{sep_form,bibliography}
\end{document}